\documentclass[11pt]{amsart}

\usepackage[ansinew]{inputenc}
\usepackage[a4paper,dvips]{geometry}
\usepackage{latexsym}
\usepackage{amssymb}
\usepackage{bbm}
\usepackage{color}
\usepackage{tikz}
\usetikzlibrary{matrix,calc}

\usepackage{graphicx}
\newtheorem{theorem}{Theorem}[section]
\newtheorem{corollary}[theorem]{Corollary}
\newtheorem{question}[theorem]{Question}
\newtheorem{lemma}[theorem]{Lemma}

\newtheorem{definition}[theorem]{Definition}

\theoremstyle{remark}
\newtheorem{remark}{Remark}

\newcommand{\NN}{\mathbb{N}}

\newcommand{\RR}{\mathbb{R}}

\newcommand{\Sy}{\mathfrak{S}}

\title[Point charges and the van der Corput sequence]{Greedy energy minimization can count in binary: Point charges and the van der Corput sequence}

\author{Florian Pausinger}
\address{School of Mathematics \& Physics, Queen's University Belfast, BT7 1NN, Belfast, United Kingdom.}
\email{f.pausinger@qub.ac.uk}

\subjclass{Riesz energy, energy minimization, van der Corput sequence, Leja sequence, universality}

\date{21/1/2020}

\begin{document}

\maketitle


\begin{abstract}
This paper establishes a connection between a problem in Potential Theory and Mathematical Physics, arranging points so as to minimize an energy functional, and a problem in Combinatorics and Number Theory, constructing 'well-distributed' sequences of points on $[0,1)$. 
Let $f:[0,1] \rightarrow \mathbb{R}$ be (i) symmetric $f(x) = f(1-x)$, (ii) twice differentiable on $(0,1)$, and (iii) such that $f''(x)>0$ for all $x \in (0,1)$. 
We study the greedy dynamical system, where, given an initial set $\{x_0, \ldots, x_{N-1}\} \subset [0,1)$, the point $x_N$ is obtained as
$$ x_{N} = \arg\min_{x \in [0,1)} \sum_{k=0}^{N-1}{f(|x-x_k|)}.$$
We prove that if we start this construction with the single element $x_0=0$, then all arising constructions are permutations of the van der Corput sequence (counting in binary and reflected about the comma): \textit{greedy energy minimization recovers the way we count in binary.} This gives a new construction of the classical van der Corput sequence. 
The special case $f(x) = 1-\log(2 \sin(\pi x))$ answers a question of Steinerberger.
Interestingly, the point sets we derive are also known in a different context as Leja sequences on the unit disk.
Moreover, we give a general bound on the discrepancy of any sequence constructed in this way for functions $f$ satisfying an additional assumption.

\end{abstract}

\section{Introduction and main result}

\subsection{A Problem in Mathematical Physics.}  A classical question in Mathematical Physics, sometimes known as Thomson's problem \cite{thomson}, is the following: suppose you have $N$ electrons on $\mathbb{S}^2$ interacting via a Coulomb potential, what are the stable equilibria? We can associate to any set of $N$ points a notion of energy
$$ E(\left\{x_0, \dots, x_{N-1}\right\}) = \sum_{i \neq j}{\frac{1}{\|x_i - x_j\|}}$$
and the main questions are then: (i) what is the minimal energy? (ii) what are configurations of points attaining minimal energy? and (iii) what do these configurations look like? These old questions are far from answered: the minimal energy is known to have an asymptotic expansion, the first few terms are known; see e.g. \cite{saffneu3, brauchart, brauch, saff2, saff}. Relatively little is known about minimal energy configurations, they seem to arrange themselves in a hexagonal pattern which is one instance of the crystallization conjecture \cite{bl, radin, theil}. The problem is notoriously difficult: the special case of $N=5$ points was only very recently solved by R. Schwartz \cite{schwartz}. For some special values of $N$ there are constructions that are known; we refer to the seminal work of Cohn-Kumar \cite{cohn}. Whenever there are topological obstructions to triangular lattice, the defects seem to localize in scars \cite{many}.
Needless to say, many of these questions remain interesting and many of the results carry over to the case of a more general domain instead of $\mathbb{S}^2$ and a more general energy functional
$$ E(\left\{x_0, \dots, x_{N-1}\right\}) = \sum_{i \neq j}{f(\|x_i - x_j\|)},$$
where $f:\mathbb{R} \rightarrow \mathbb{R}$. The choice $f(x) = x^{-s}$, Riesz potentials, is among the most popular.
The only case where the problem is known to have been widely solved is the case
of the one-dimensional torus: for a rather large class of functions the optimal arrangement
is known to be equispaced points \cite{brauchart3}.

Given the difficulty of these questions, one would assume that it is much more difficult to study the dynamical problem, where one starts with a given set of points $\{x_0, \ldots, x_{N-1}\} \subset [0,1)$ (or possibly just a single point) and then defines the next element in the sequence in a greedy fashion via
\begin{equation} \label{functional} x_{N} = \arg\min_{x \in [0,1)} \sum_{k=0}^{N-1}{f(|x-x_k|)}, \end{equation}
i.e. adding the point in the location of the minimum of the energy (with the caveat that should the minimum not be unique, then any choice is admissible).  One of the main contributions of our paper is that for a wide class of functions $f$, this dynamical system can indeed be studied for one-periodic functions, i.e. on the one-dimensional torus $\mathbb{T}$ identified with the unit interval. Moreover, it gives rise to a surprisingly rigid dynamical structure.

\subsection{A Problem in Combinatorics and Number Theory.} \label{sec:problem} Suppose we want to distribute a sequence of points $X=(x_n)_{n=1}^{\infty}$ evenly over $[0,1)$ in the most regular fashion -- how would one do it?  If we knew in advance that we want to place $N$ points, then we would presumably place them at equidistant intervals. However, what if we wanted the sequence to be distributed regularly at all times (i.e. also at all intermediate stages)? We now make this notion precise and define the extreme discrepancy
of the first $N$ points of $X$ as
$$ D_N(X) = \sup_{0 \leq \alpha<\beta < 1} \left |  \frac{\# \{ 1 \leq i \leq N: \alpha\leq x_i < \beta\}}{N} - (\beta-\alpha) \right |.$$
It is easy to see that $D_N \geq 1/N$. How small can it be? In particular, is there a sequence $X$ such that $D_N(X) \leq c N^{-1}$ for all $N \in \mathbb{N}$? This question, originally due to van der Corput, was answered by van Aardenne-Ehrenfest who showed that no such result exists; see \cite{KN} for details. The problem was finally solved by Schmidt \cite{sch72} who proved that $D_N \geq c  N^{-1} \log{N}$ for infinitely many $N \in \mathbb{N}$; see \cite{survey1, KN, pausinger1, sch72} and references therein for smallest known constants. Sequences with $D_N(X) \leq c  N^{-1} \log{N}$ for a constant $c>0$ and all $N$ are called \emph{low discrepancy} sequences. In particular, this bound matches classical constructions of sequences, one of the most famous of which is the van der Corput sequence (\cite{vdc35a, vdc35b}):
$$ 0, \frac12, \frac14, \frac34, \frac18, \frac58, \frac38, \frac78, \frac{1}{16}, \frac{9}{16}, \frac{5}{16}, \dots$$ 
The original definition of the van der Corput sequence is based on writing integers in binary expansion: the $n-$th element is given by 
\begin{enumerate}
\item[(i)] writing the integer in binary, i.e. $22 = 10110_2$
\item[(ii)] inverting the order of the digits $10110 \rightarrow 01101$
\item[(iii)] writing a comma in front of it and interpreting it as a real number in $[0,1]$
$$ x_{22} = .01101_2 =\frac{1}{4} + \frac{1}{8} + \frac{1}{32} = \frac{13}{32}.$$
\end{enumerate}
The van der Corput sequence and its various generalisations are known to be very close to optimal with regards to discrepancy; see \cite{faure1, survey1, pausinger1}.

\subsection{A possible connection.} There are very few known constructions of low discrepancy sequences, most of them are based on structures from Number Theory or Combinatorics. 
Low discrepancy sequences are important in numerical integration since they ensure smallest possible approximation errors; see \cite{DP}. While the theory is very well understood in one dimension, there are to-date two different conjectures of the sharp lower bound on the discrepancy of arbitrary sequences in more dimensions; \cite{DP}. In particular, it is not clear whether there exist sequences more regular than anything we can construct so far. 

Motivated by this, Steinerberger \cite{stefan1, stefan2} recently proposed to study whether regular sequences could be constructed via dynamical systems of the type outlined in \eqref{functional}.
More precisely, suppose we are given $\{x_0, \ldots, x_{N-1} \} \subset [0,1)$, then he proposed to construct $x_N$ in a greedy manner as
\begin{equation} \label{eqn:stefan} 
x_N = \arg \min_{\min_k |x-x_k| \geq N^{-10}} \sum_{k=0}^{N-1} 1 - \log(2 \sin (\pi |x-x_k|)), 
\end{equation}
and if the minimum is not unique, any choice is admissible.
Steinerberger proves that, independently of the initial conditions, such sequences satisfy $D_N \leq c N^{-1/2} \log{N}$. Moreover, dropping the technical gap condition in \eqref{eqn:stefan}, he conjectures:
\begin{enumerate}
\item[(A)] (weak form): the discrepancy of every such sequence satisfies $$D_N \leq c N^{-1} \log{N};$$
\item[(B)] (strong form): and the constant $c$ does not depend on the initial set from which this construction is started.
\end{enumerate}
We give affirmative answers to these conjectures in Theorem \ref{thm:main} in the special case when we are given only one initial point $x_0$.
It was conjectured that if we start with $x_0 = 0$, then the van der Corput sequence is an admissible choice for the greedy selection rule; see Figure \ref{fig:vdc10}. The main purpose of our paper is to prove that this is indeed the case and that it holds at a much larger level of generality for dynamical systems of this type; for a precise statement of this result see Theorem \ref{thm:main}.
Moreover, we give a general discrepancy bound in Theorem \ref{thm:main2} for a subset of the functions considered in Theorem \ref{thm:main} and arbitrary initial choices which is slightly worse than Steinerberger's bound for the special case of 
$f(x)=1 - \log(2 \sin (\pi |x|))$.

\begin{figure}[h!]
\begin{center}
\includegraphics[scale=0.8]{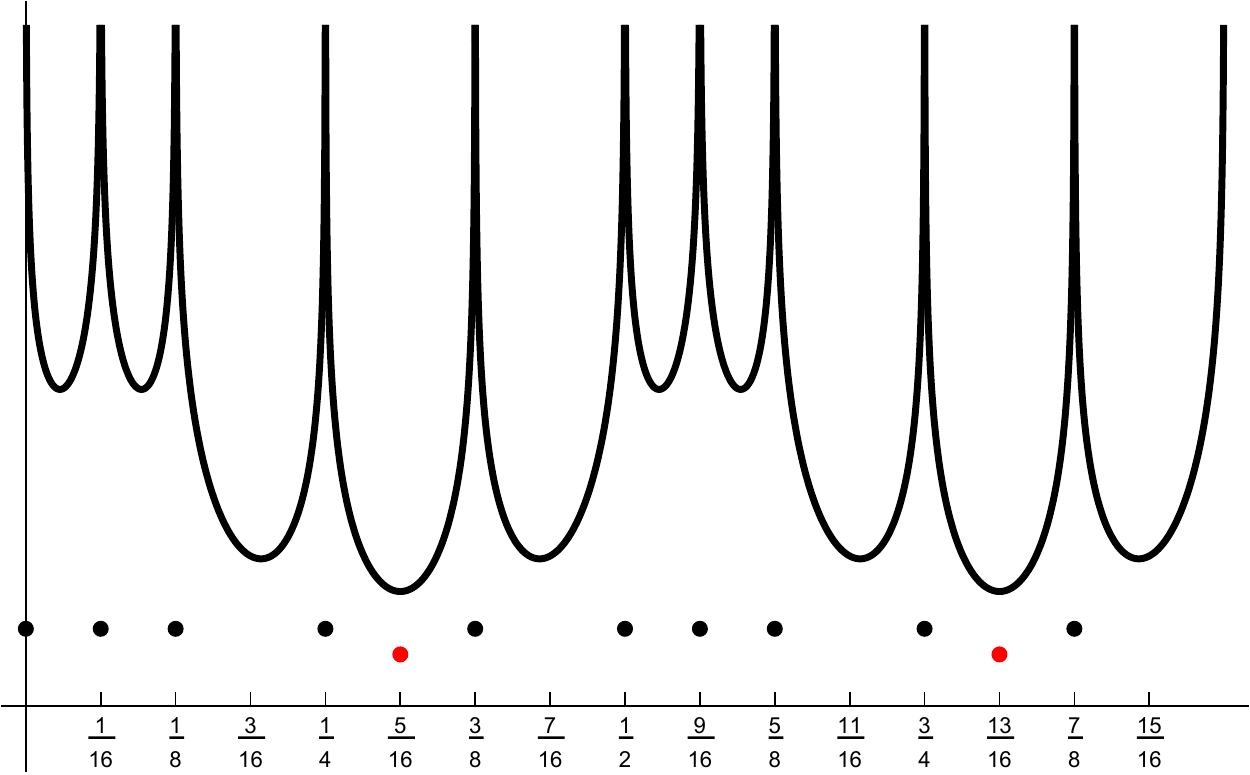}
\end{center}
\caption{The set $M_{10}$ of the first 10 points (black) of the classical van der Corput sequence and the function  $\sum_{k=0}^9 1 - \log(2 \sin(\pi |x-x_k|))$. The red points are the candidates for the next point as suggested by the greedy algorithm; note that the $x_{10}=.0101_2=5/16$.} \label{fig:vdc10}
\end{figure}

\subsection{Outline.} 
We state our main results, Theorem \ref{thm:main} and Theorem \ref{thm:main2}, in Section \ref{sec:result}; in particular this section also contains in Subsection \ref{sec:family} the definition of the sequences used to characterise the output of the greedy algorithm of Steinerberger as well as an outline of our proof strategy.
In Section \ref{sec:vdc} we review properties of the van der Corput sequence. We also recall important generalisations and the method of Faure to calculate the discrepancy of permuted van der Corput sequences. 
In Section \ref{sec:family2} we study the discrepancy of the particular generalised van der Corput sequences used to characterise the output of the algorithm. As a main result we show in Theorem \ref{thm:family} resp. in Corollary \ref{cor:family} that all such sequences have the same discrepancy. In Section \ref{sec:algorithm} we relate our results to the greedy algorithm. We show in Theorem \ref{thm:vdc} that the classical van der Corput sequence is an admissible output of the algorithm and we finally prove Theorem \ref{thm:main}. Section \ref{sec:main2} contains a proof of Theorem \ref{thm:main2}.

\section{Main results}
\label{sec:result}

\subsection{Preliminaries}
Let $\mathfrak{S}_b$ be the set of all permutations of $\{0,1,\ldots, b-1\}$, i.e. all permutations of the first $b$ non-negative integers.
Throughout this paper we sometimes write a concrete permutation $\sigma$ as tuple, e.g. $\sigma=(0,2,4,1,3)$ meaning that $0$ is mapped to $0$, $1$ is mapped to $2$, $2$ is mapped to $4$ and so on; i.e. the number at the $i$-th position of the tuple is the image of $i$ under $\sigma$. Moreover, for tuples of numbers $\tau_1,\tau_2$ of length $b_1, b_2$ we write $(\tau_1,\tau_2)$ for the concatenation of the two tuples; i.e. $(\tau_1,\tau_2)$ is a tuple of numbers of length $b_1+b_2$. Finally, for permutations (or tuples) of fixed length $b$, we are sometimes only interested in its initial segment of the first $k$ numbers. In this case we may write $\sigma=(\tau_k,\ast)$ meaning that the first $k$ elements of the permutation are determined by $\tau_k$ and the rest of the permutation can be any arrangement of the remaining $b-k$ numbers.

\subsection{Permuted van der Corput sequences} \label{sec:Defvdc}
Let $a_j(n)$ denote the $j$th coefficient in the binary representation 
$$ n = a_0 \cdot 2^0+ a_1 \cdot 2^1 + a_2 \cdot 2^2+ \ldots = \sum_{j=1}^{\infty} a_j 2^j $$
of an integer $n$, in which $0 \leq a_j (n) \leq 1$ and if $0\leq n < 2^m$, then $a_j(n)=0$ for all $j\geq m$. The binary radical inverse function is defined as $S_2: \NN_0 \rightarrow [0,1)$,
$$ S_2(n) = \frac{a_0(n)}{2} + \frac{a_1(n)}{2^2} +\frac{a_2(n)}{2^3} + \ldots = \sum_{j=0}^{\infty} \frac{a_j(n)}{2^{j+1}}.$$
Then the classical van der Corput sequence is defined as $\mathcal{S}_2=(S_2(n))_{n\geq 0}$.

Faure \cite{faure1} generalised the definition of the classical van der Corput sequence in two ways. First, he replaced the binary representation of an integer by its general $b$-adic representation for a fixed integer base $b\geq 2$. This allows for the definition of the $b$-adic radical inverse function $S_b$, which in turn can be used to define van der Corput sequences in general bases; i.e. $\mathcal{S}_b=(S_b(n))_{n\geq 0}$. 
Furthermore, for $\sigma \in \mathfrak{S}_b$ Faure defines the generalised (or permuted) van der Corput sequence $\mathcal{S}_b^{\sigma}=(S_b^{\sigma}(n))_{n \geq 0}$ for a fixed base $b\geq 0$ via the permuted $b$-adic radical inverse function
$$ S_b^{\sigma}(n) = \sum_{j=0}^{\infty} \frac{\sigma(a_j(n))}{b^{j+1}}. $$
Every such sequence is uniformly distributed modulo 1; \cite[Propri\'et\'e 3.1.1]{faure1}.

\subsection{Family of permutations} \label{sec:family}
We inductively define a set of permutations $\mathcal{P}_m \subset \mathfrak{S}_{2^m}$ in each basis $b_m=2^m$. We start with $b_1=2$ and $\mathcal{P}_1=\{(0,1)\}$ and we obtain the set $\mathcal{P}_{m+1}$ from $\mathcal{P}_m$ in the following way: 
We first multiply each permutation $\sigma \in \mathcal{P}_m$ with 2 and denote the resulting tuple of numbers as $2 \sigma$ and the set of all such tuples as $2 \mathcal{P}_m$. Next, each $2\sigma \in 2 \mathcal{P}_m$ gives rise to $2^m$ new tuples: For each odd $a$ with $1\leq a \leq 2^{m+1}$ we form a new tuple $2\sigma \oplus a$ by adding $a$ to $2\sigma$ (addition modulo $2^{m+1}$). The set of all such tuples is denoted by $2\mathcal{P}_m \oplus a$. Finally, the set $\mathcal{P}_{m+1}$ is defined as the set of all permutations $(2\sigma, 2\sigma' \oplus a)$ for $\sigma, \sigma' \in \mathcal{P}_m$ and odd $a$ with $1\leq a \leq 2^{m+1}$.

As examples we construct $\mathcal{P}_2$ and $\mathcal{P}_3$. First, we have that $2\mathcal{P}_1 =\{(0,2)\}$ and $2\mathcal{P}_1\oplus a =\{(1,3),(3,1)\}$. Consequently $\mathcal{P}_2=\{(0,2,1,3),(0,2,3,1)\}$. Next, we have 
$$2\mathcal{P}_2 =\{(0,4,2,6), (0,4,6,2)\},$$ 
and 
$$2\mathcal{P}_2 \oplus a =\{(1,5,3,7), (3,7,5,1),(5,1,7,3), (7,3,1,5), (1,5,7,3),(3,7,1,5),(5,1,3,7),(7,3,5,1)\}$$
from which we can build $\mathcal{P}_3$.


\subsection{Main result and open questions}
We can now state our first main result.
\begin{theorem}[Counting in Binary] \label{thm:main} Let $f:[0,1] \rightarrow \mathbb{R}$, be a function satisfying 
\begin{enumerate}
\item[(i)] $f(x) = f(1-x)$;
\item[(ii)] $f$ is twice differentiable on $(0,1)$;
\item[(iii)] $f''(x) >0$ for all $x \in (0,1)$;
\end{enumerate}
and let the sequence $X=(x_N)_{N=0}^{\infty} \subset [0,1)$ be defined by $x_0 = 0$ and
\begin{equation*}  x_{N} =  \arg\min_{x \in (0,1)} \sum_{k=0}^{N-1}{f(|x-x_k|)}, \end{equation*}
for $N>0$ and  where every global minimum is admissible if it happens not to be unique. Then for every fixed $N>0$ there exist an $m>0$ with $N\leq 2^m$ and a $\sigma \in \mathcal{P}_m$ such that $x_k= S_{2^m}^{\sigma}(k)$ for all $1\leq k \leq N$ and
$D_N(X) = D_N(\mathcal{S}_{2^m}^{\sigma})=D_N(\mathcal{S}_2)$.
\end{theorem}

\begin{remark}
Note that our assumptions on $f$ imply that $f$ is a strictly convex function on $[0,1]$; but not every strictly convex function necessarily satisfies $f''(x) >0$ for all $x \in (0,1)$. However, every uniformly convex function, i.e. $f$ such that $f''(x)\geq \theta >0$, satisfies our assumptions.
\end{remark}

\begin{remark}
The point sets we obtain are closely related to Leja sequences on the unit disk. Let $K$ be a nonempty compact subset of the complex plane. A sequence $(a_n)_{n \geq 0}$ of points in $K$ is a Leja sequence for $K$ if the following extremal metric property holds for $N\geq 1$,
$$  \prod_{n=0}^{N-1} | (a_N-a_n) | = \max_{z \in K} \prod_{n=0}^{N-1} | (z-a_n) |.$$
Thus, the $(N+1)$st term $a_N$ of a Leja sequence must maximise the product of distances to the $N$ previous ones; see \cite{leja1,leja2}. 
There are very few classes of compact sets for which Leja sequences are explicitly known. However, if $K$ is the unit disk then the structure of the corresponding Leja sequences are known and described in \cite[Theorem 5 and Corollary 2]{leja1}.
Importantly, the maximum principle implies that these points lie on the boundary of $K$; i.e. on the unit circle if $K$ is the unit disk.

Note that Leja sequences are defined via an explicit functional which is studied on different compact sets. In our setting, we fix the domain and vary the functionals. It is interesting that our minimising sequences have the same structure as the particular Leja sequences on the unit disk. 
\end{remark}

\begin{remark}
We believe Theorem \ref{thm:main} is a strong indicator that the dynamical version 
$$ x_{N} =  \arg\min_x \sum_{k=0}^{N-1}{f(|x-x_k|)},$$
of the static equilibrium problem in mathematical physics might give rise to interesting structures. In the one-dimensional case, it certainly connects in a very substantial way to structures in number theory. To emphasize this, we explicitly state that:
\begin{quote}
Greedy energy minimization on the one-dimensional torus automatically recovers the way we count in binary.
\end{quote}
\begin{question}
Which structures arise when we use different classes of function that are not covered by the theorem? Can we get similar structures? Or totally different dynamics?
\end{question}
\end{remark}
\begin{remark} The result has nontrivial implications for the study of uniform distribution. It starts by providing a novel definition of the van der Corput sequence; i.e. start the greedy algorithm \eqref{functional} with $\{0\}$ and always pick the smallest of the suggested minima.  The second statement in the main Theorem, i.e. Discrepancy being preserved over all possible choices, shows that potential theoretic approaches along the lines of what was proposed by Steinerberger \cite{stefan1, stefan2} might indeed have intimate ties to discrepancy. 
\end{remark}
Theorem \ref{thm:main} tells us that we can use symmetric and uniformly convex functions to construct the van der Corput sequence via greedy energy minimisation.
\begin{question}
Which functions $f$ and seeds $\{x_0, \ldots, x_{N-1}\}$ generate low discrepancy sequences via greedy energy minimisation as in \eqref{functional}? 
\end{question}
\begin{question}
Is there another family of functions which can be used to reconstruct permuted van der Corput sequences in prime bases $b>2$ via greedy energy minimisation?
\end{question}
\begin{question}
We have verified Steinerberger's conjectures for the case of one initial element. What about the case of an arbitrary set $\{x_0, x_1\}$ of two (or more) points in $[0,1]$?
\end{question}
We suspect that the classical van der Corput sequence and its permutations studied in this paper form in a way a unique (and natural) link between sequences constructed via greedy energy minimisation and sequences obtained from traditional methods in number theory.

\subsection{Outline of proof.} \label{sec:proofStrategy}
We prove Theorem \ref{thm:main} in 2 steps. In the first step (i.e. in Section \ref{sec:family2}) we calculate the discrepancy of permuted van der Corput sequences $\mathcal{S}_{2^m}^{\sigma}$, $\sigma \in \mathcal{P}_m$. We conclude in Corollary \ref{cor:family} that all such sequences have the same extreme discrepancy as the classical van der Corput sequence $\mathcal{S}_2$. Our argument uses the general machinery of Faure for the calculation of the discrepancy of permuted van der Corput sequences which we recall in Section \ref{sec:vdc} and is based on various symmetries exhibited by the permutations in $\mathcal{P}_m$.

In the second step (i.e. in Section \ref{sec:algorithm}) we relate our particular family of permuted van der Corput sequences to the greedy algorithm of Steinerberger. We show that every such sequence can be obtained from the algorithm (Lemma \ref{lem:admissible2}) and that every output of the algorithm can be described by such a sequence (Lemma \ref{lem:admissible}). Thus, we obtain a full characterisation of the possible outputs of the greedy algorithm and by the results of the first step we also know the discrepancy of any such output.

This concludes the proof of Theorem \ref{thm:main}.

\subsection{General discrepancy bound.}
Finally, we conclude with a general discrepancy bound for sequences generated from the greedy algorithm and an arbitrary initial set of values. We assume wlog that $f$ has mean 0 and require two additional assumptions on $f$. We obtain the discrepancy bound for functions $f$ that are bounded and whose Fourier coefficients $\hat{f}(k)$ defined as
$$\widehat{f}(k) = \int_0^1 f(x) \mathrm{e}^{-2 \pi i k x} \ dx $$
satisfy a condition on their decay.

\begin{theorem} \label{thm:main2}
Let $f:[0,1] \rightarrow \RR$ be as in Theorem \ref{thm:main}, assume in addition that $f$ is bounded and normalised to have mean value 0, and that
$$\hat{f}(k) \geq c |k|^{-2}$$ for a constant $c>0$ and all $k \neq 0$.
Define the sequence $X=(x_N)_{N=0}^{\infty} \subset [0,1)$ via an initial set $\{x_0, \ldots, x_{m-1}\} \subset [0,1)$ and the rule 
\begin{equation*}  x_{N} =  \arg\min_{x \in (0,1)} \sum_{k=0}^{N-1}{f(|x-x_k|)}, \end{equation*}
for $N\geq m$, where every global minimum is admissible if it happens not to be unique. Then this sequence satisfies
$$ D_N(X) \leq \frac{\tilde{c}}{N^{1/3}}, $$
in which the constant $\tilde{c} >0$ only depends on the initial set.
\end{theorem}




\section{Properties of the van der Corput sequence}
\label{sec:vdc}

Faure \cite{faure1} developed a comprehensive machinery to calculate the discrepancy of permuted van der Corput sequences as defined in Section \ref{sec:Defvdc}.
Our proof of Theorem \ref{thm:main} makes use of this machinery; i.e. we use Faure's method to show that all permutations $\sigma \in \mathcal{P}_m$ generate permuted van der Corput sequences with the same discrepancy as the classical van der Corput sequence in base 2.
To facilitate our proof we recall the main properties of van der Corput sequences as well as Faure's method in this section. Our main result of this section is Lemma \ref{lem:intr2} in which we relate the classical van der Corput sequence to our family of permutations $\mathcal{P}_m$, $m\geq 1$.


Let $[\alpha,\beta) \subseteq [0,1)$.
For $N\geq 1$ and $X=(x_n)_{n\geq 0}$, let 
$$A([\alpha,\beta),N,X):\# \{ 0 \leq i \leq N-1: \alpha \leq x_i < \beta\}$$ 
denote the number of indices $0 \leq i\leq N-1$ for which $\{x_i\} \in [\alpha,\beta)$.
An infinite sequence $X$ is \emph{uniformly distributed modulo 1} (u.d. mod 1) if
\begin{equation} \label{udt1}
\underset{N\rightarrow \infty} \lim \frac{A([\alpha,\beta), N, X) }{N} = \beta - \alpha, 
\end{equation}
for every subinterval $[\alpha,\beta)$ of $[0,1)$. The uniform distribution property \eqref{udt1} of an infinite sequence is usually quantified with one of several different notions of discrepancy. 
We put 
$$R([\alpha,\beta), N, X)=A([\alpha,\beta),N,X)- (\beta-\alpha) N,$$
and obtain for the extreme discrepancy, $D_N(X)$, of the first $N$ points of $X$
$$N D_N(X)=\underset{[\alpha,\beta) \subseteq \mathrm{I}} \sup | R([\alpha,\beta), N, X) |.$$
Note that Faure does not divide the discrepancy by $N$; i.e. his formulas for the discrepancy (e.g. in \cite{faure1}) give $N D_N(X)$ in our notation.



\subsection{Self similarity of van der Corput sequences. }
The classical van der Corput sequence and its permuted generalisations are self similar as shown in the following two lemmas. Let 
$$
M_N := \{ S_2(i) : 0 \leq i \leq N-1 \}, 
$$
denote the set of the first $N$ points of the van der Corput sequence and let
$$ M_{N_1,N_2} := M_{N_2} \setminus M_{N_1} $$
for $N_2>N_1$ denote a segment of the van der Corput sequence of length $N_2-N_1$.

\begin{lemma} \label{selfsimilar}
Let $m_k, m_{k-1}, \ldots, m_1$ be a decreasing sequence of non-zero integers. Let $N_1=\sum_{j=1}^{k} 2^{m_j}$ and let $N_2 = N_1 + 2^{m_0}$ for $m_0\leq m_1$. Then $N_2-N_1 = 2^{m_0}$ and 
$$ M_{N_1, N_2} = M_{2^{m_0}} + \sum_{j=1}^k \frac{1}{2^{m_j + 1}}. $$
\end{lemma}

\begin{proof}
We have that
\begin{align*} 
M_{N_1, N_2} &= \{ S_2(N_1), \dots, S_2(N_2 -1) \} \\ 
&= \left \{ \sum_{i=0}^{m_0} \frac{a_i(n)}{2^{i+1}} + \sum_{j=1}^k \frac{1}{2^{m_j +1}} : N_1 \leq n \leq N_2 -1 \right\} \\
&= \left \{ \sum_{i=0}^{m_0} \frac{a_i(n)}{2^{i+1}} : 0 \leq n \leq 2^{m_0}-1 \right \} + \sum_{j=1}^k \frac{1}{2^{m_j +1}} \\
&=M_{2^{m_0}} + \sum_{j=1}^k \frac{1}{2^{m_j + 1}}.
\end{align*}
\end{proof}

We can repeat the calculation from the proof of Lemma \ref{selfsimilar} to obtain the following general version for
$$M_N^{\sigma} := \{ S_b^{\sigma}(i) : 0 \leq i \leq N-1,\  \sigma \in \mathfrak{S}_b \}, \ \ \ \text{ and } \ \ \ M_{N_1,N_2}^{\sigma} := M_{N_2}^{\sigma} \setminus M_{N_1}^{\sigma}, $$
for $N_2>N_1$.
\begin{lemma} \label{selfsimilar2}
Let $m_k, m_{k-1}, \ldots, m_1$ be a decreasing sequence of non-zero integers, $b\in \NN$, $b\geq 2$ and $\sigma \in \Sy_b$ with $\sigma(0)=0$. Let $N_1=\sum_{j=1}^{k} a_{m_j} b^{m_j}$ and let $N_2 = N_1 + b^{m_0}$ for $m_0\leq m_1$. Then $N_2-N_1 = b^{m_0}$ and 
$$ M^{\sigma}_{N_1, N_2} = M_{b^{m_0}} + \sum_{j=1}^k \frac{\sigma(a_{m_j}(N_1))}{b^{m_j + 1}}. $$
\end{lemma}



\subsection{Discrepancy of permuted van der Corput sequences}
In the following we introduce Faure's system of basic functions which can be used to calculate the discrepancy of generalised van der Corput sequences. The theory of Faure was applied to the study of various point sets and sequences over the last 40 years and is very well explained and illustrated in the recent survey \cite{survey1}.
The exact formulas for the discrepancy of generalised van der Corput sequences are based on a set of elementary functions which are defined for any permutation $\sigma \in \Sy_b$. Let 
$$\mathcal{X}_b^{\sigma}:= \left( \frac{\sigma(0)}{b}, \frac{\sigma(1)}{b},\ldots ,\frac{\sigma(b-1)}{b} \right).$$ 
For $h \in \{0,1,\ldots, b-1\}$ and $x \in \left[\frac{k-1}{b},\frac{k}{b}\right[$, in which $1\leq k \leq b$ is an integer, we define $$\varphi_{b,h}^{\sigma}(x):=\left\{
\begin{array}{ll}
A([0,h/b[,k,\mathcal{X}_b^{\sigma})-h x & \mbox{ if } 0 \le h \le \sigma(k-1),\\[5pt]
(b-h)x- A([h/b,1[,k,\mathcal{X}_b^{\sigma}) & \mbox{ if }\sigma(k-1)< h < b.
\end{array}\right.$$
The functions $\varphi_{b,h}^{\sigma}$ are continuous on $[0,1[$, piecewise affine and are extended to the real numbers by periodicity.
Moreover, we define
\begin{align*}
\psi_b^{\sigma,+} = \max_{0 \leq h < b} \varphi_{b,h}^{\sigma}, \ \ \ \ \psi_b^{\sigma,-} = \max_{0 \leq h < b} (- \varphi_{b,h}^{\sigma}), \ \ \ \
\psi_b^{\sigma}=\psi_b^{\sigma,+}+\psi_b^{\sigma,-}.
\end{align*}
For an infinite, one-dimensional sequence $X$ we set
$$ N D_N^+ = \sup_{0 \leq \alpha \leq1} R( [0,\alpha], N, X ), \ \ \ \text{ and } \ \ \  N D_N^- = \sup_{0 \leq \alpha \leq1} (- R( [0,\alpha], N, X )).$$
Then we get from \cite[Th\'{e}or\`{e}me 1]{faure1} for $N\geq 1$ that
\begin{align*} \label{formula:DN}
N D_N^+(\mathcal{S}_{b}^{\sigma}) = \sum_{j=1}^{\infty} \psi_b^{\sigma,+} (N/b^j), \ \   & N D_N^-(\mathcal{S}_{b}^{\sigma}) = \sum_{j=1}^{\infty} \psi_b^{\sigma,-} (N/b^j), \ \
N D_N(\mathcal{S}_{b}^{\sigma}) = \sum_{j=1}^{\infty} \psi_b^{\sigma} (N/b^j).
\end{align*}
Note that the infinite series in these formulas can indeed be computed exactly; for details we refer to \cite[Section~3.3.6, Corollaire 1]{faure1}. 
The exact formulas can be used for the asymptotic analysis of the discrepancy of generalised van der Corput sequences. By \cite[Th\'{e}or\`{e}me 2]{faure1}, 
$$
s(\mathcal{S}_b^{\sigma}) =\limsup_{N \rightarrow \infty} \frac{N D_N(\mathcal{S}_b^{\sigma}) }{\log N} = \frac{\alpha_b^{\sigma}}{ \log b}
$$
with
\begin{equation} \label{asym}
\alpha_b^{\sigma} = \underset{m\geq 1} \inf \ \underset{x \in \RR}  \sup \left( \frac{1}{m} \sum_{j=1}^m \psi_b^{\sigma} (x/b^j) \right).
\end{equation}
Furthermore, we introduce the function
\begin{equation} \label{Fn}
F_m^{\sigma}(x)= \sum_{j=0}^{m-1} \psi_b^{\sigma} (x b^j),
\end{equation}
and rewrite \eqref{asym} as
$$ \alpha_b^{\sigma} = \inf_{m\geq 1}  \left( \max_{x \in [0,1]} F_m^{\sigma}(x)/m  \right). $$
Note that the local maxima of $\psi_b^{\sigma}$ have arguments of the form $x=k/b$ for $k\in \NN$ with $0\leq k \leq b-1$. As shown in \cite[Lemme 4.2.2]{faure1}, the sequence $(\max_{x \in [0,1]} F_m^{\sigma}(x)/m)_{m\geq 1}$ is decreasing. In particular,
\begin{equation} \label{asymSeq}
\alpha_b^{\sigma} \leq \ldots \leq \ \max_{x \in [0,1]} F_2^{\sigma}(x)/2 \ \leq \  \max_{x \in [0,1]} F_1^{\sigma}(x) =  \max_{x \in \RR} \ \psi_b^{\sigma},
\end{equation}
with $\alpha_b^{\sigma} = \lim_{m \rightarrow \infty} \max_{x \in [0,1]} F_m^{\sigma}(x)/m$.

In  \cite[Th\'{e}or\`{e}me 6]{faure1} the asymptotic constants for van der Corput sequences that are generated from the identity permutation in base $b$ are calculated:
$$ s(\mathcal{S}_b) = \left\{
\begin{array}{ll}
\frac{b-1}{4 \log b} & \mbox{ if $b$ is odd,}\\
\frac{b^2}{4 (b+1) \log b}& \mbox{ if $b$ is even. }
\end{array} \right.$$
Consequently, the only sequence generated from an identity permutation that improves the result of the classical sequence in base $2$ is $\mathcal{S}_3$; for more information see the recent survey \cite{pausinger1}.

\subsection{Symmetries}

The Symmetry as well as the Swapping Lemma facilitate the analysis of structurally similar permutations and state that a shift, reflection or reversal of a permutation does not change the discrepancy of the generated sequence. This is not surprising since we interpret $[0,1]$ as a circle.

\begin{lemma}[Symmetry Lemma]  \label{lem:SameDisc}
Let $0<a<b$ be an integer, let $\sigma \in \Sy_b$ and let $\sigma', \sigma'' \in \Sy_b$ be defined as
$$ \sigma'(x) = \sigma(x) + a \pmod{b}, \ \ \ \text{ and } \ \ \  \sigma''(x) = - \sigma(x) \pmod{b}$$
for $0\leq x \leq b-1$. Then it holds for all $N$ that
\begin{equation*} 
D_N(\mathcal{S}_b^{\sigma}) = D_N(\mathcal{S}_b^{\sigma'}) \ \ \ \text{ and } \ \ \ D_N(\mathcal{S}_b^{\sigma}) = D_N(\mathcal{S}_b^{\sigma''}).
\end{equation*}
\end{lemma}
The first part of the lemma was noted in \cite[Th\'eor\`eme 4.4]{cf93}; the second assertion was discussed in \cite[Lemma~2.1]{patop18}.

\begin{lemma}[Swapping Lemma] \label{lem:swap}
Let $\sigma, \mu_b \in \Sy_{b}$ with $\mu_b(k)=b-k-1$ for $0 \leq k \leq b-1$ being the swapping permutation, then we have that
\begin{align*}
\psi_{b}^{\sigma \circ \mu_b, +} &= \psi_b^{\sigma,-} \\
\psi_{b}^{\sigma \circ \mu_b, -} &= \psi_b^{\sigma,+}.
\end{align*}
\end{lemma}
This was shown in \cite[Lemme 4.4.1]{faure1}.

\subsection{Intrication}

Faure defined an operation \cite[Section~3.4.3]{faure1} which takes two arbitrary permutations $\sigma, \tau$ in bases $b$ and $c$ and outputs a new permutation, $\sigma \cdot \tau$ in base $b\cdot c$. The motivation for this definition comes from the following property which was first noted in \cite[Proposition 3.4.3]{faure1}.

\begin{lemma}[Intrication] \label{intr}
For $\sigma \in \Sy_b$ and $\tau \in \Sy_c$ define $\sigma \cdot \tau \in \Sy_{bc}$ as
$$ \sigma \cdot \tau (k'' b+ k' ) = c \, \sigma(k') + \tau(k''),$$
for $0\leq k' < b$ and $0 \leq k'' < c$. Then,
\begin{equation*} 
\psi_{bc}^{\sigma \cdot \tau}(x) = \psi_{b}^{\sigma}(c x) + \psi_{c}^{\tau}(x),
\end{equation*}
such that 
$$\max_{x \in \RR} \ \psi_{bc}^{\sigma \cdot \tau}(x) \leq \max_{x \in \RR} \ \psi_{b}^{\sigma}(x) + \max_{x \in \RR} \ \psi_{c}^{\tau}(x).$$
\end{lemma}

\begin{remark} If we set $\sigma=\tau$, then the intrication $\sigma \cdot \sigma$ gives a permutation in base $b^2$ whose $\psi$-function is the function $F_2^{\sigma}$ defined in \eqref{Fn}. 
In this special case the new permutation generates the same sequence as the original permutation.
\end{remark}

In particular, we see that repeated intrications of the permutation $(0,1)$ generates permutations in bases $b_m=2^m$ of a particular form:
In the notation of Lemma \ref{intr}, let $b=c=2$ and let $\sigma=\tau=(0,1)$. Then we get
\[ \begin{array}{c|c|cr}
(k'',k') & 2 k'' + k' & 2 \sigma(k') + \tau(k'') \\
\hline
(0,0) & 0 & 2\cdot 0 + 0 = 0 \\
(0,1) & 1 & 2\cdot 1 + 0 = 2 \\
(1,0) & 2 &2\cdot 0 + 1 = 1 \\
(1,1) & 3 &2\cdot 1 + 1 = 3 
\end{array}\]
If we take now the resulting permutation $(0,2,1,3)$ in base $b_2=4$ and form the intrication with $(0,1)$ then we see that the new permutation in base $8$ is obtained from two copies of the permutation in base $4$; i.e. setting $\sigma_2:=\sigma\cdot \tau$ we can write the new permutation as
$$ \sigma_3= \sigma_2 \cdot (0,1) = (2 \sigma_2, 2\sigma_2 \oplus 1). $$
Using Lemma \ref{intr} it is easy to see that this holds in general:
\begin{lemma} \label{lem:intr2}
If $\sigma_1=(0,1)$ and $\sigma_m := \sigma_{m-1}\cdot (0,1)$ for $m\geq 2$, then 
$ \sigma_m=(2 \sigma_{m-1}, 2 \sigma_{m-1}\oplus 1). $ In particular,
$$\psi_{2^m}^{\sigma_m} (x)= F_m^{\sigma_1} (x)$$
for all $x\in [0,1]$.
\end{lemma}

Hence, every permutation $\sigma_m$ is contained in the set $\mathcal{P}_m$ which we defined in Section \ref{sec:family}.

\section{A family of permutations}
\label{sec:family2}

In this section we study the discrepancy of sequences generated from permutations in $\mathcal{P}_m \subset \mathfrak{S}_{2^m}$. The main result of this section is the observation that all the functions $\psi_{2^m}^{\sigma}$ are identical for $\sigma \in \mathcal{P}_m$. Thus, using Lemma \ref{lem:intr2}, we see that all permutations in $\mathcal{P}_m$ generate permuted van der Corput sequences with the same asymptotic discrepancy constant (in fact, with the same discrepancy) as the classical van der Corput sequence. We summarise the main observations of this section in Corollary \ref{cor:family}.

We start with an important technical observation.
\begin{lemma}\label{lem:doubling}
Let $b=2^m$ and let $1\leq k \leq 2^m$. If $\psi_{2^m}^{\sigma}=\psi_{2^m}^{\sigma'}$ on $[0,1]$ for $\sigma, \sigma' \in \mathcal{P}_m$, then $\psi_{2^{m+1}}^{(2 \sigma,\ast)}=\psi_{2^{m+1}}^{(2 \sigma',\ast)}$ on $[0,1/2]$ for the tuples $2\sigma, 2\sigma' \in 2  \mathcal{P}_m$.
\end{lemma}
\begin{proof}
Let $h \in \{0,1, \ldots, 2^{m+1}-1 \}$ and let $x\in [(k-1)/2^{m+1}, k/2^{m+1}[$ in which $1 \leq k \leq 2^m$ is an integer.
Note that the tuples $2 \sigma$ and $2 \sigma'$ suffice to calculate $\psi_{2^{m+1}}^{(2 \sigma,\ast)}$ and $\psi_{2^{m+1}}^{(2 \sigma',\ast)}$ on $[0,1/2]$.
To prove the result we compare the functions $\varphi_{2^{m+1},h}^{(2\sigma,\ast)}$ to the corresponding function $\varphi_{2^{m},\tilde{h}}^{\sigma}$, in which $0\leq h \leq 2^{m+1}$ and $0 \leq \tilde{h} \leq 2^m$.

{\bf Case 1: Even $h$. } If $h=2 \tilde{h}$ is even then 
$$ \varphi_{2^{m+1},h}^{(2\sigma,\ast)} \left( \frac{x}{2} \right) = \varphi_{2^{m},\tilde{h}}^{\sigma} \left( x \right)$$
for $x \in [0,1[$.
To see this, we look at the definition of the $\varphi$-functions and observe that 
$$A\left( \left[0,h/2^{m+1} \right[,k, \mathcal{X}_{2^{m+1}}^{(2\sigma,\ast)} \right) = A\left( \left[0,\tilde{h}/2^{m}\right[,k, \mathcal{X}_{2^{m}}^{\sigma}  \right),$$
and $hx/2= 2 \tilde{h}x/2=\tilde{h}x$.

{\bf Case 2: Odd $h$. } If $h=2 \tilde{h}-1$ is odd then
$$ \varphi_{2^{m+1},h}^{(2\sigma,\ast)} \left( \frac{x}{2} \right) = \varphi_{2^{m},\tilde{h}}^{\sigma} \left( x \right) + \frac{x}{2}$$
for $x \in [0,1[$.
Follows from a similar argument.

Consequently, the $\varphi$-functions change independently of the particular $\sigma$ when going from $m$ to $m+1$ and hence if the related $\psi$ functions where identical for $\sigma, \sigma'$, they remain identical.
\end{proof}

\begin{definition}
Two sets $\tau, \tau'$ are $m$-inverse if they are disjoint and their union is the set $\{0,1,\ldots, 2^m-1\}$.
\end{definition}

To illustrate this definition, observe that the two sets $\{0,1,2,4,6\}$ and $\{3,5,7\}$ are $3$-inverse.
Being $m$-inverse implies an important equality. We formulate the next lemma in the specialised way in which we will use it later.
\begin{lemma}\label{lem:BWsymmetry}
Let $b=2^{m+1}$, let $0\leq k \leq 2^m-1$ and let $\tau_k$ and $\tau_{2^m-k}$ be two tuples of length $k$ and $2^m-k$.
If $\tau_{2^m-k}$ and $\{0,2, \ldots, 2^m-2\} \cup \tau_k$ are $(m+1)$-inverse, then for $\xi \in [0,1]$
$$ \psi_{2^{m+1}}^{(\tau_{2^m-k}, \ast)} \left( \frac{2^m - k-\xi}{2^{m+1}} \right) = \psi_{2^{m+1}}^{(\{0,2,\ldots, 2^m-2 \} \cup \tau_k, \ast)}  \left( \frac{2^m + k+\xi}{2^{m+1}} \right). $$
\end{lemma}
\begin{proof}
Let $N_1=2^m -k$, $N_2 = 2^m +k$ and let $y \in \{1/2^{m+1}, 2/2^{m+1}, \ldots, 1 \}$. Then $N_1 + N_2 = 2^{m+1}$. Since $\tau_{2^m-k}$ and $\{0,2, \ldots, 2^m-2\} \cup \tau_k$ are $(m+1)$-inverse we have that
$$ A([0,y[, N_1, \mathcal{X}_{2^{m+1}}^{(\tau_{2^m-k}, \ast)}) + A([0,y[, N_2, \mathcal{X}_{2^{m+1}}^{(\{0,2,\ldots, 2^m-2 \} \cup \tau_k, \ast)}) = 2^{m+1} y. $$
From this we get that
\begin{align*}
A([0,y[, N_1, \mathcal{X}_{2^{m+1}}^{(\tau_{2^m-k}, \ast)}) - (N_1-\xi) y &= 
(2^{m+1} y - A([0,y[, N_2,\mathcal{X}_{2^{m+1}}^{(\{0,2,\ldots, 2^m-2 \} \cup \tau_k, \ast)}) ) - (2^{m+1} - N_2 -\xi) y\\
& = - A([0,y[, N_2,\mathcal{X}_{2^{m+1}}^{(\{0,2,\ldots, 2^m-2 \} \cup \tau_k, \ast)}) + (N_2+\xi) y.
\end{align*}
Hence, all corresponding $\varphi$ functions on $\left[  \frac{2^m+k}{2^{m+1}}, \frac{2^m+k+1}{2^{m+1}}\right]$ are reflections of the according $\varphi$-functions on $\left[ \frac{2^m-k-1}{2^{m+1}}, \frac{2^m-k}{2^{m+1}} \right]$ with opposite signs. 
Consequently, the corresponding $\psi$ functions restricted to the same intervals are reflections of each other with respect to the line $x=1/2$.
\end{proof}


To prepare for the proof of the main theorem of this section, we make one more technical observation, which we state again in the special form in which we will use it below.

\begin{lemma} \label{lem:swap2}
Let $\sigma \in \mathcal{P}_m$ and let $\sigma':=\sigma \circ \mu_{2^m}$, in which $\mu_{2^m}$ is the swapping permutation. Let $\tau_k$ denote the first $k$ elements of $2 \sigma \oplus a$ and $\tau'_{2^m-k}$ denote the first $2^m-k$ elements of $2 \sigma' \oplus a$ for odd $1 \leq a \leq 2^{m+1}$ and addition modulo $2^{m+1}$. Then $\tau_k \cup \{0,2, \ldots, 2^{m+1}-2 \}$ and $\tau'_{2^m-k}$ are $(m+1)$-inverse.
\end{lemma}

\begin{proof}
Note that the swapping permutation simply reverses the elements of $\sigma$; i.e. $$\sigma(i)=(\sigma \circ \mu_{2^m})(2^m-i-1)$$ for all $0\leq i \leq 2^m-1$. Hence, the last $2^m-k$ elements of $\sigma$ are exactly the first $2^m-k$ elements of $\sigma'$.
\end{proof}

Lemma \ref{lem:BWsymmetry} and Lemma \ref{lem:swap2} imply that for all $\sigma \in \mathcal{P}_{m}$ the function $\psi_{2^m}^{\sigma}(x)$ for $x\in [1/2,1]$ is a reflection of $\psi_{2^m}^{\sigma}(x)$ with $x \in [0,1/2]$; i.e. the graph of $\psi_{2^m}^{\sigma}(x)$ with $x \in [0,1/2]$ is reflected at the line $x=1/2$.








\begin{theorem} \label{thm:family}
For $m\geq 1$ let $\sigma, \sigma' \in  \mathcal{P}_m$. Then $\psi_{2^m}^{\sigma}(x) = \psi_{2^m}^{\sigma'}(x)$ for all $x \in [0,1]$.
\end{theorem}

\begin{proof}
We prove the theorem by induction on $m$. The theorem is trivially true for $m=1$. Now let $m=2$. It is easy to see that both permutations give rise to the same $\psi$-function. To turn to the induction step, we assume that the assertion is true up to $m$. In particular this means that all permutations in the set $ \mathcal{P}_m$ generate permuted van der Corput sequences with identical $\psi$-function. 

By Lemma \ref{lem:doubling} all tuples in $2  \mathcal{P}_m$ have identical $\psi$-functions for $x \in [0,1/2]$. 
By Lemma \ref{lem:SameDisc} adding a constant modulo $2^m$ to the permutations in $ \mathcal{P}_m$ does not change the $\psi$-function. Hence, also all tuples in $2  \mathcal{P}_m \oplus a$ have identical $\psi$-functions for $x \in [0,1/2]$.
Furthermore, by Lemma \ref{lem:swap} swapping a permutation does not change its $\psi$-function. Hence, our considerations also apply to all swapped permutations $\sigma \circ \mu_{2^m}$.

Now take arbitrary $\sigma, \tilde{\sigma} \in \mathcal{P}_m$ and build a permutation $(2 \sigma, 2\tilde{\sigma} \oplus a) \in \mathcal{P}_{m+1}$.
We already know that the $\psi$-functions of all such permutations are identical for $x \in [0,1/2]$.
Now let $\tilde{\tau}_k$ denote the first $k$ elements of $2\tilde{\sigma} \oplus a$.
By Lemma \ref{lem:swap2} we know that $\tilde{\tau}_k \cup \{0,2,\ldots, 2^m-2\}$ is $(m+1)$-inverse to $\tilde{\tau}'_{2^m-k}$. Hence, we can apply Lemma \ref{lem:BWsymmetry} to calculate the value of 
$$\psi_{2^{m+1}}^{(\{0,2,\ldots, 2^m-2 \} \cup \tau_k, \ast)}  \left( x \right),$$
for $x \in [\frac{2^m + k}{2^{m+1}},\frac{2^m + k+1}{2^{m+1}}[$.
However, we know that the swapped permutation $\tilde{\sigma} \circ \mu_{2^m}$ has the same $\psi$-function on $[0,1/2]$ as every other permutation in $\mathcal{P}_m$.
Thus, we obtain the values $\psi_{2^{m+1}}^{(2 \sigma, 2\tilde{\sigma} \oplus a)}$ on $[1/2,1]$ via the reflection of the $\psi$ function on $[0,1/2]$ of any other tuple in $2\mathcal{P}_m$.

Since we picked an  arbitrary permutation in $\mathcal{P}_{m+1}$ we conclude that every permutation in $\mathcal{P}_{m+1}$ has the same $\psi$-function.
\end{proof}

\begin{corollary} \label{cor:family}
For $m\geq 1$ with $\sigma_m$ as defined in Lemma \ref{lem:intr2} and $\sigma \in  \mathcal{P}_m$ we have $\psi_{2^m}^{\sigma}(x) = \psi_{2^m}^{\sigma_m}(x)$ for all $x \in [0,1]$, such that for all $N\geq 1$
$$ D_N(\mathcal{S}_{2^m}^{\sigma})=D_N(\mathcal{S}_2). $$
\end{corollary}


\section{Relation to algorithm and proof of Theorem \ref{thm:main}}
\label{sec:algorithm}

In the following we relate the results of Section \ref{sec:vdc} and \ref{sec:family2} to the greedy algorithm and prove Theorem \ref{thm:main}.
\begin{definition}
Let $A>0$ be a constant. A real-valued function $f:[0,A] \rightarrow \RR$ is admissible if (i) it is twice differentiable on $(0,A)$; (ii) $f(x)=f(A-x)$ for $x \in [0,A]$; (iii) $f''(x)>0$ for all $x \in (0,A)$. 
\end{definition}

\begin{lemma} \label{lem:function}
Let $f$ be an admissible function.
Then $g(x):= f(x) + f(A/2+x)$ for $x \in [0,A/2]$ attains a unique minimum at $x=A/4$ and $g$ as a function on $[0,A/2]$ has the same properties as $f$ on $[0,A]$; i.e. $g$ is admissible on $[0,A/2]$.
\end{lemma}

\begin{proof}
We have that $g'(x)= f'(x) + f'(A/2+x)$. Hence, $g'(x)=0$ if and only if $f'(x) = -f'(A/2+x)$. Since $f(x)=f(A-x)$ we have that $f'(x)=-f'(A-x)$ from which we see that $g'(A/4)=0$. Since $f'$ is increasing, we have that $g''$ is positive and therefore $x=A/4$ is a minimum of $g$.

Moreover, for $x \in [0,A/2]$, we have that 
$$g(x)=f(x) + f(A/2+x) = f(A-x) + f(A-A/2-x)=g(A/2-x)$$ 
and since $f'$ is increasing, also $g'$ is increasing on $[0,A/2]$.
\end{proof}

Due to symmetry we can get a copy of $g$ in the interval $[A/2,A]$.
\begin{definition}
Given an admissible function $f$ on $[0,1]$ and a point set $\mathcal{P}=\{ x_0, \ldots, x_{N-1}\} \subset [0,1]$ we define
$$f_k(x,\mathcal{P}):=f( |x-x_k| ),$$
for all $0\leq k \leq N-1.$
If $\mathcal{P}$ is clear from the context we simply write $f_k(x)$.
\end{definition}
The functional studied by Steinerberger \cite{stefan1, stefan2} is an example of an admissible function; for $x\in [0,1]$
$$ \tilde{f}(x) := 1 - \log( 2 \sin (\pi x) );$$
see Figure \ref{fig:functions} for an illustration.

\begin{figure}
\includegraphics[scale=0.4]{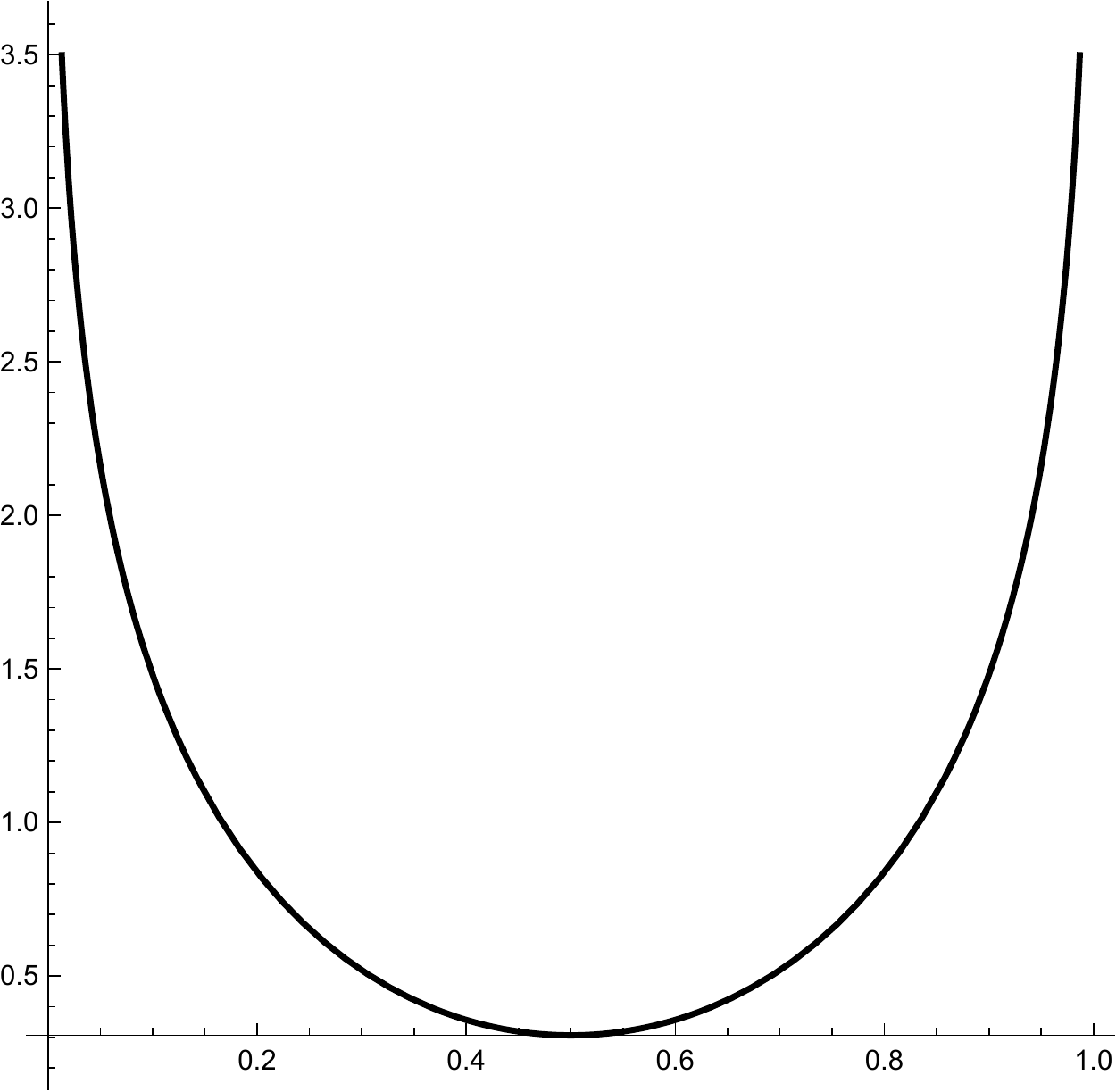} \quad \quad
\includegraphics[scale=0.4]{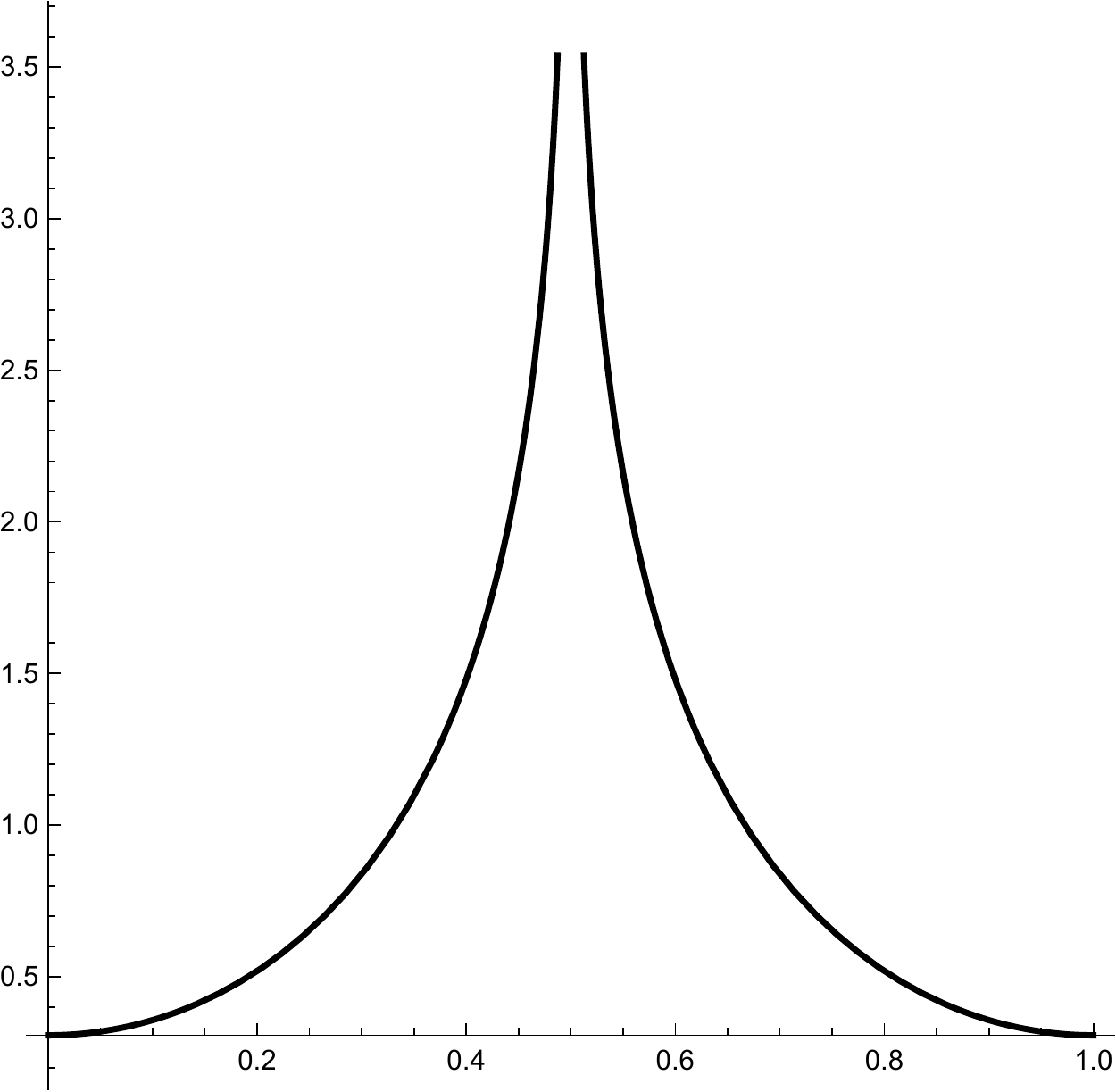}
\caption{The function $\tilde{f}_0$, the function $\tilde{f}_1$ with $x_1=1/2$.} \label{fig:functions}
\end{figure}

\begin{lemma} \label{minima}
Let $f$ be an admissible function on $[0,1]$.
Let $N=2^m$ and let $M_N=\{S_2(0), \dots, S_2(N-1)\}=\{ x_0, \ldots, x_{N-1}\}$ be the set of the first $N$ points of the classical van der Corput sequence. Define
$$ G_m (x):= \sum_{k=0}^{N-1} f_k (x). $$
Then $G_m(x)$ has $2^m$ minima at points of the form $ i/2^m + 1/2^{m+1}$ for $0\leq i \leq N-1$. Moreover, $G_m$ restricted to intervals $[j/2^m, (j+1)/2^m]$ is again admissible.
\end{lemma}

\begin{proof}
We prove the Lemma by induction on $m$. Direct calculation shows that the claim is true for $m=1$; i.e. for the set $M_2=\{ 0,1/2\}$.

Assume that the claim is true for $m$. Then the function $G_m$ defined via the set $M_{2^m}$ attains $2^m$ minima at points  $i/2^m + 1/2^{m+1}$. Now it is easy to see that shifting all points in $M_{2^m}$ by a constant, simply shifts the graph of the corresponding function $G_m$ but does not change its qualitative behaviour since we are working on the one-dimensional torus; see Figure \ref{fig:functions}.
So we can define a second set 
$$\tilde{M}_{2^m} = M_{2^m} + 1/2^{m+1} = \{S_2(2^m), \ldots, S_2(2^{m+1}-1) \}.$$ 
By our assumption the corresponding function $\tilde{G}_m$ has again $2^m$ minima which all lie at points of the form $i/2^m$ for $i\leq 0 < 2^m$; i.e. the minima of $G_m$ are the points in $\tilde{M}_{2^m}$ and the minima of $\tilde{G}_m$ are the points of $M_{2^m}$; see Figure \ref{fig:functions2} for an illustration.

Next, we define on every interval $[j/2^{m+1}, (j+1)/2^{m+1}]$ the function $G_{m+1}(x):= G_m(x) + \tilde{G}_m(x)$. Applying Lemma \ref{lem:function} in each interval $[j/2^{m+1}, (j+1)/2^{m+1}]$ proves the claim.
\end{proof}

\begin{figure}
\includegraphics[scale=0.5]{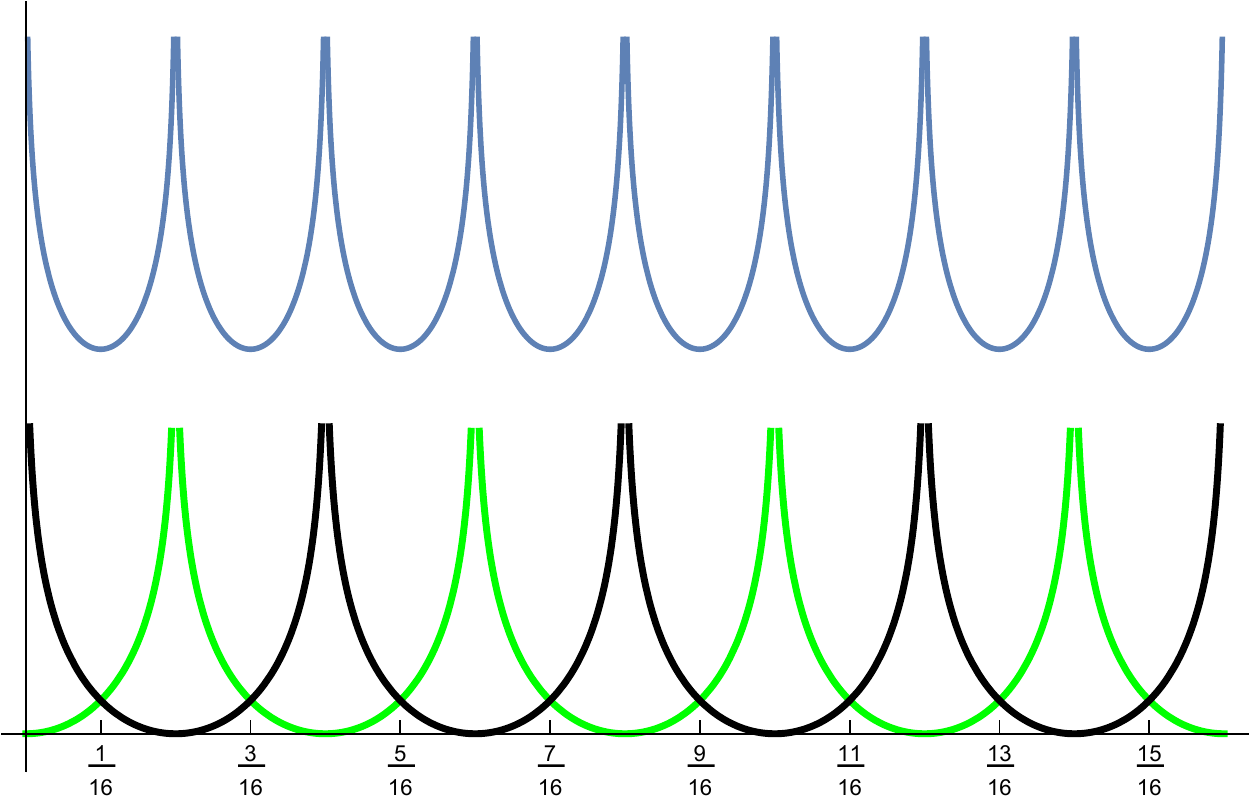} 
\caption{The functions $G_2$ (black), $\tilde{G}_2$ (green) and $G_3$ (blue).} \label{fig:functions2}
\end{figure}

\begin{theorem} \label{thm:vdc}
For arbitrary $N$ let $M_N = \{S_2(0), \dots, S_2(N-1)\}=\{ x_0, \ldots, x_{N-1} \}$ be the set of the first $N$ points of the van der Corput sequence. If $f$ is an admissible function on $[0,1]$ and if we use $M_N$ as a seed for the greedy algorithm defined via \eqref{functional}, then the $(N+1)$st point of the van der Corput sequence is among the minima.
\end{theorem}

\begin{proof}
We write $N=N_1=\sum_{j=1}^{k} 2^{m_j}$ in its binary representation such that $m_k>m_{k-1}>\ldots >m_1$ is a decreasing sequence of non-zero integers. According to this representation we split the set $M_N$ into disjoint subsets such that there is one subset $M^j$ containing $2^{m_j}$ points for every $m_j$; in particular we write
\begin{align*} 
M_N &=M^k \cup M^{k-1} \cup \ldots \cup M^1 \\
&=M_{N_k} \cup M_{N_{k-1}, N_k} \cup \ldots \cup M_{N_1, N_2}
\end{align*}
in which $N_i := \sum_{j=i}^k 2^{m_j}$.
Due to the self similarity of the van der Corput sequence (see Lemma \ref{selfsimilar}) we get a shifted version $\tilde{G}_{m_j}$ of the function $G_{m_j}$ for each set $M^j$; see Figure \ref{fig:Nelf} and Figure \ref{fig:Nelf2}.
To conclude our proof we show that for each $k>j\geq 1$ the set of minima of $\tilde{G}_{m_j}$ is contained in the set of minima of the function $\tilde{G}_{m_{j+1}}$.
From Lemma \ref{selfsimilar} we know that 
\begin{align*}
M^j =M_{N_j, N_{j+1}} =M_{2^{m_j}} + \sum_{h=j+1}^k \frac{1}{2^{m_h+1}}
\end{align*}
and $M^{j+1} =M_{N_{j+1}, N_{j+2}}$.
Furthermore from Lemma \ref{minima} we know that the minima of $\tilde{G}_{m_j}$ are at the points
\begin{equation}\label{minJ} \frac{i}{2^{m_j}} + \frac{1}{2^{m_j+1}} + \sum_{h=j+1}^k \frac{1}{2^{m_h+1}}, \end{equation}
whereas the minima of $\tilde{G}_{m_{j+1}}$ are at
\begin{equation}\label{minJJ} \frac{i}{2^{m_{j+1}}} + \frac{1}{2^{m_{j+1}+1}} + \sum_{h=j+2}^k \frac{1}{2^{m_h+1}}. \end{equation}
Now it is easy to see that all arguments of minima of the form \eqref{minJ} are also of the form \eqref{minJJ}.

The set of arguments of global minima of the resulting function $H_N=\sum_{h=0}^{N-1} f_h$ for $M_N$ is then the intersection of the set of all arguments of minima of the functions $\tilde{G}_{m_j}$. In particular, this set contains exactly the $2^{m_1}$ points of the form
\begin{equation}\label{minJJJ} 
\frac{i}{2^{m_1}} + \frac{1}{2^{m_1+1}} + \sum_{h=2}^k \frac{1}{2^{m_h+1}} = \frac{i}{2^{m_1}} + \sum_{h=1}^k \frac{1}{2^{m_h+1}}, \end{equation}
It follows from the definition of the van der Corput sequence that the $(N+1)$st point of the sequence is among them. 
\end{proof}

\begin{figure}
\includegraphics[scale=0.35]{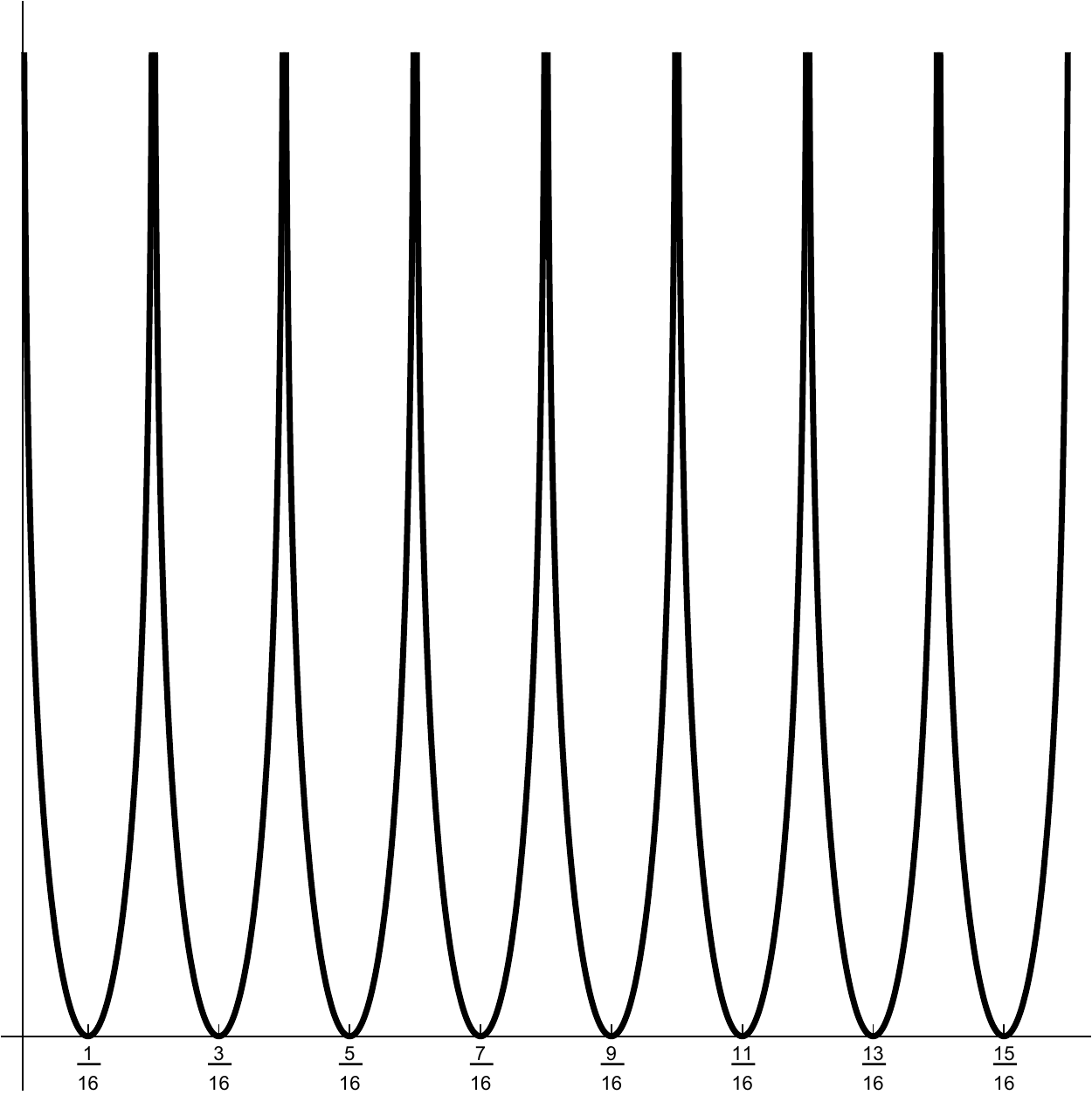} \quad \quad
\includegraphics[scale=0.35]{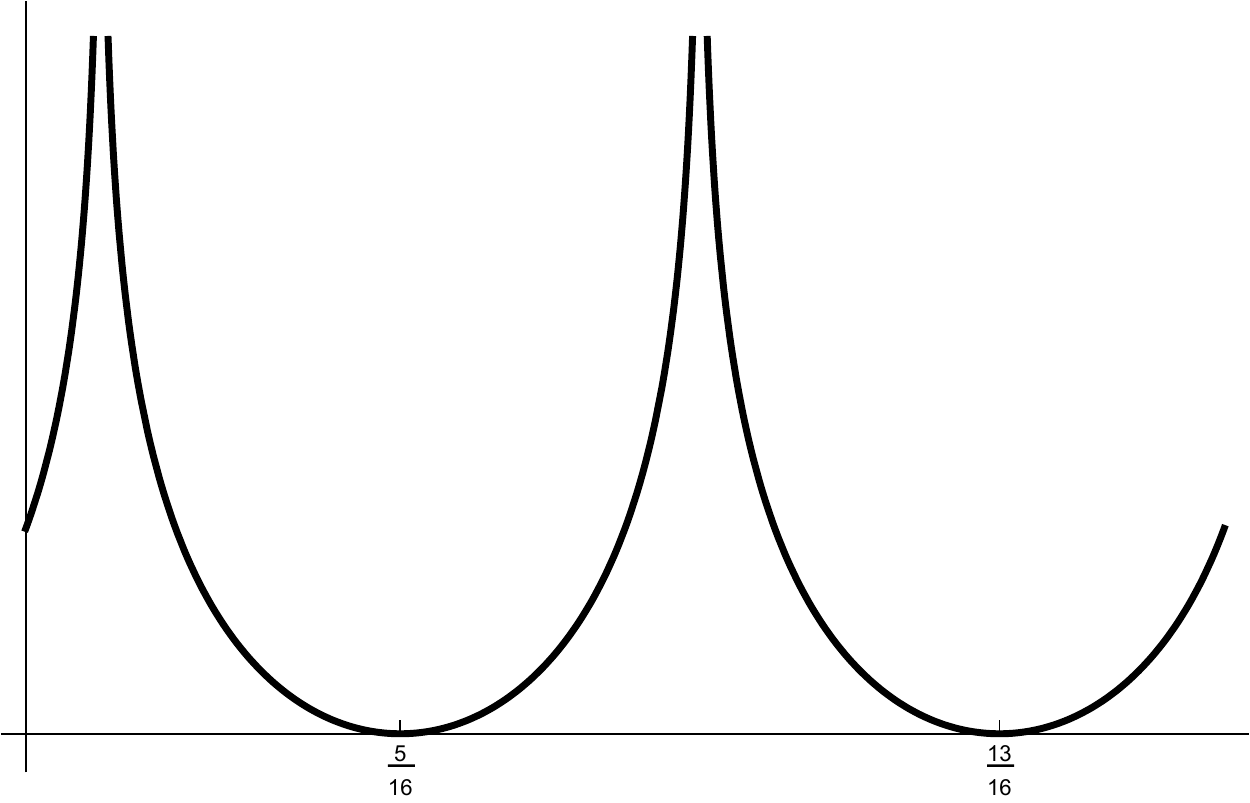} \quad \quad
\includegraphics[scale=0.35]{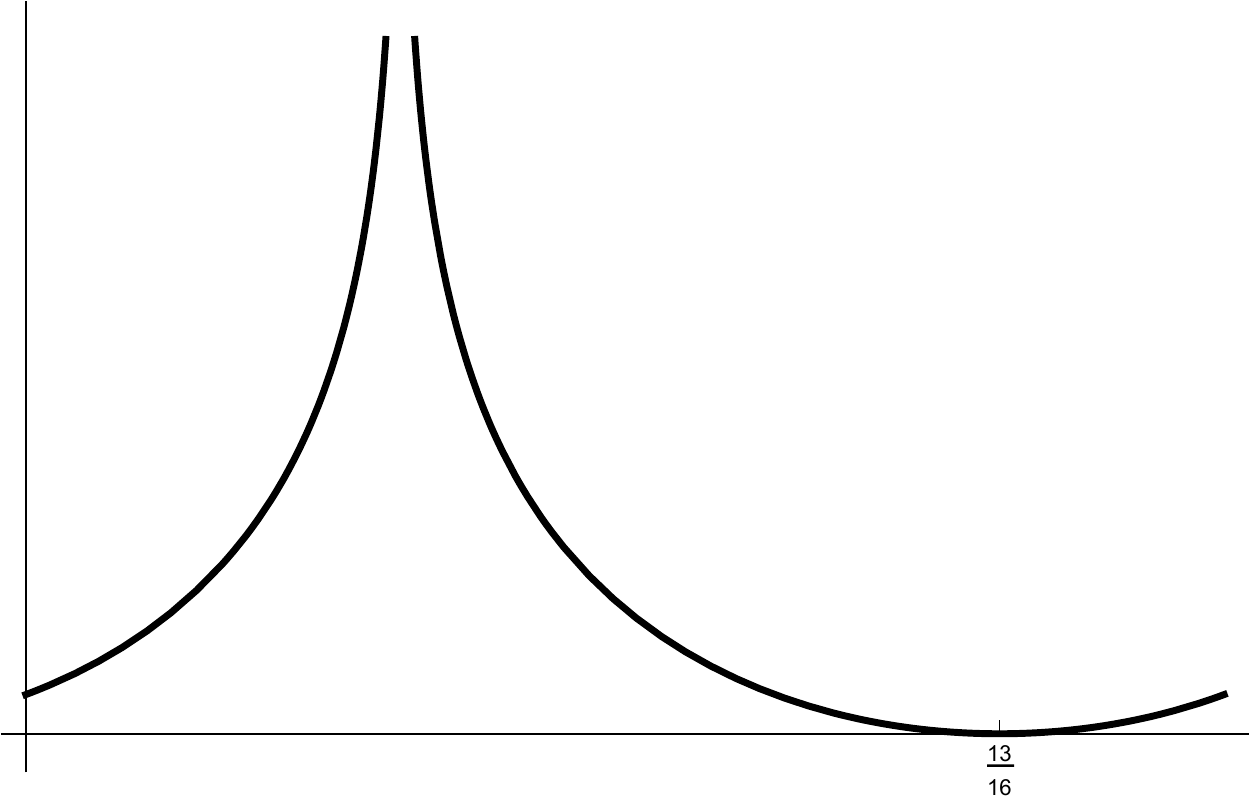} 
\caption{Example $N=11=2^3+2^1+2^0$. Functions of the set $M^3=\{0, 1/8,\ldots, 7/8 \}$, $M^1=\{ 1/16, 9/16\}$ and $M^0=\{ 5/16\}$.} \label{fig:Nelf}
\end{figure}

\begin{figure}
\includegraphics[scale=0.35]{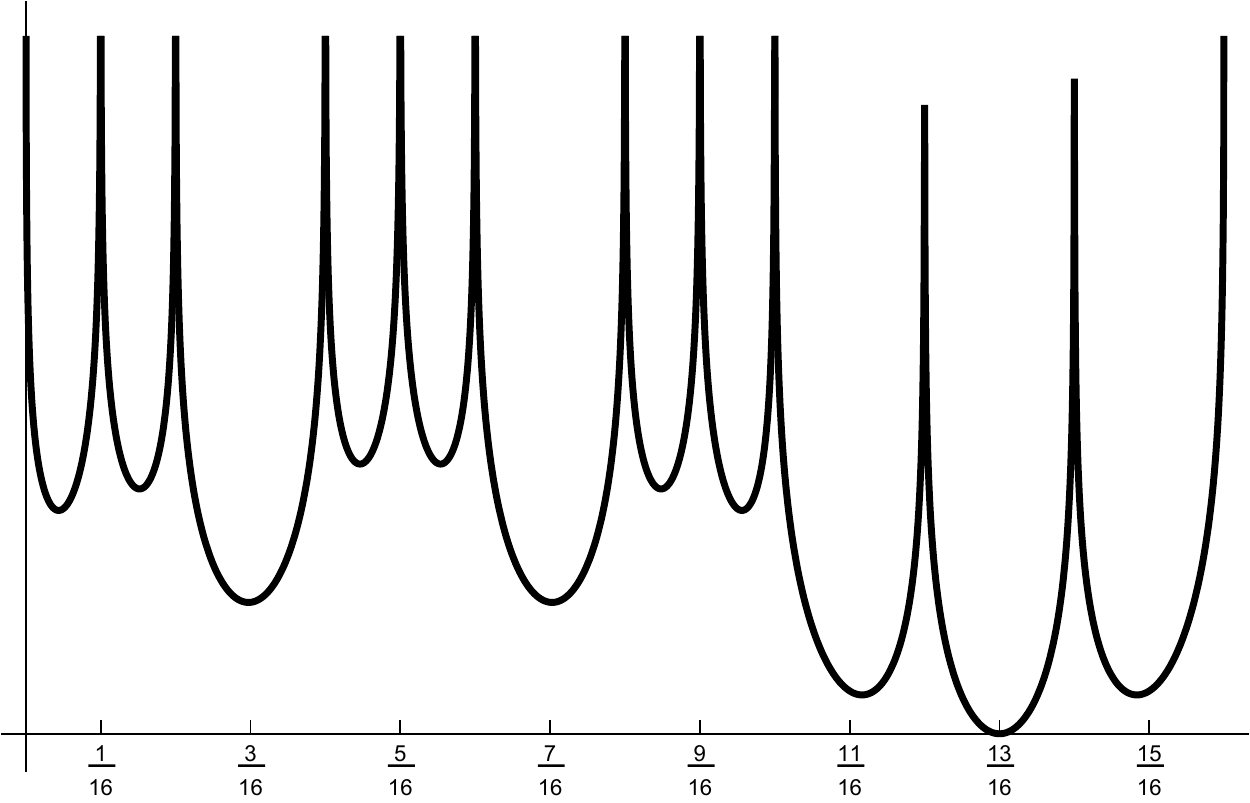} 
\caption{The function $H_{11}$.} \label{fig:Nelf2}
\end{figure}

Finally, we observe that every sequence generated by a permutation in $\mathcal{P}_m$ can be realised by the algorithm (when starting with $\{0\}$) and that for a fixed $2^{m-1}< N \leq 2^{m}$ every output of the algorithm can be found as initial segment of a sequence $\mathcal{S}_{2^m}^{\sigma}$ with $\sigma \in \mathcal{P}_m$. 
\begin{definition}
We call a set $\{x_0, x_1, \ldots, x_{N-1} \}$ of $N$ points \emph{admissible} if it can be obtained as output of the algorithm \eqref{functional} for an admissible function $f$ on $[0,1]$ when starting with the seed $\{0\}$.
\end{definition}

\begin{lemma} \label{lem:admissible}
If $N\geq 2$ and $\{x_0, x_1, \ldots, x_{N-1} \}$ is admissible then there exists an $m$ with $2^{m-1}< N \leq 2^m$ and a $\sigma \in \mathcal{P}_m$ with $x_i = S_{2^m}^{\sigma}(i)$ for all $0 \leq i \leq N-1$.
\end{lemma}

\begin{proof}
Starting with $x_0=0$ and $m=1$, we get $x_1=\arg\min_{x \in [0,1)} f(|x-x_k|) = 1/2$. Hence, for $N=2$ the set $\{0,1/2\}$ is the only admissible set. This set can be realised setting $m=1$ and picking $\sigma=(0,1)$ which is contained in $\mathcal{P}_1$. Thus the assertion is true for $N=2$. 
Now we move on to $m=2$. Recall that $\mathcal{P}_2$ contains exactly two permutations; i.e. $\mathcal{P}_2=\{(0,2,1,3), (0,2,3,1)\}$.
The algorithm suggests two values as third point; i.e. $x_2=1/4$ or $x'_2=3/4$. Thus, we have two admissible sets $\{0,1/2,1/4\}$ and $\{0,1/2,3/4\}$ with $3$ elements.  The first set can be realised with $\sigma=(0,2,1,3)$ whereas the second option can be realised with $\sigma'=(0,2,3,1)$. For $N=4$ the algorithm leaves no choice. The only admissible set with four points is $\{0,1/2,1/4,3/4\}$ which can be realised with both permutations in $\mathcal{P}_2$. Thus the assertion is true for $m=2$.

Now we prove the lemma by induction on $m\geq 2$. 
We assume the assertion is true for all $N \leq2^m$; i.e. for every admissible set $\{x_0, x_1, \ldots, x_{N-1} \}$ with $N \leq 2^m$ there exists a $\sigma \in \mathcal{P}_m$ with $x_i = S_{2^m}^{\sigma}(i)$ for all $0 \leq i \leq N-1$.

In particular, this means that $\mathcal{P}_m$ contains a permutation encoding every possible sequence of decisions we take when distributing the first $2^m$ points using the greedy algorithm and starting with $x_0=0$. For $N=2^m$ we always get exactly one admissible set; i.e. 
$$ \left \{0, \frac{1}{2^m}, \ldots, \frac{2^m-1}{2^m}  \right\}.$$

To make the induction step from $m$ to $m+1$, we first use Lemma \ref{minima} to see that the algorithm suggests $2^m$ minima as candidates for $x_{2^m}$, i.e. the $(N+1)$st point, which are of the form $i/2^m + 1/2^{m+1}$ for $0 \leq i \leq 2^m-1$.
Picking any of these $2^m$ minima, the important observation is that the function $H_{2^m+1} = \sum_{h=0}^{2^m} f_h$ and 
$H_{2^m+1}-H_{2^m}=\bar{H}_1 = f_{2^m}$ have the same set of arguments of minima. This can be seen as in the proof of Theorem \ref{thm:vdc}: We look at the intersection of a nested sequence of sets (the set of arguments of minima for $H_{2^m}$ contains the set of arguments of minima of $\bar{H}_1$) and this intersection does not change, i.e. is still the set of arguments of minima of $\bar{H}_1$, if we discard the outermost set.
We can repeat this argument for any of the following $N$s with $2^m+2 \leq N \leq 2^{m+1}$. In particular, we see that once we have picked the point for $N=2^m+1$ we encounter the same tree of decisions as for the first $2^m$ points, just with a different start value.

Now it is easy to see that we have indeed a permutation $(2\sigma, 2\sigma' \oplus a)$ in $\mathcal{P}_{m+1}$ for every path through the decision tree of the algorithm.
Looking at the construction of $\mathcal{P}_{m+1}$ we see that by our assumption there is a permutation, $\sigma$, in $\mathcal{P}_m$ for every initial choice of first $2^m$ points. Then we can accommodate for any choice of $(N+1)$st point via the addition of an odd $a$ modulo $2^{m+1}$. And finally, once we picked $a$, i.e. the new starting value $a/2^m + 1/2^{m+1}$ for $0 \leq a \leq 2^m-1$, we can realise our second (shifted by the constant $a$) run through the decision tree to determine the remaining $2^{m}-1$ points via an appropriate permutation $\sigma'$ in $\mathcal{P}_m$.

Thus, if the assertion holds for $m$, then it also holds for $m+1$.
\end{proof}

\begin{lemma} \label{lem:admissible2}
Let $N\geq 1$ and let $m$ be such that $2^{m-1}< N \leq 2^m$. For every $\sigma \in \mathcal{P}_m$ the set 
$$ \{S_{2^m}^{\sigma}(0), S_{2^m}^{\sigma}(1), \ldots, S_{2^m}^{\sigma}(N-1) \} $$
is admissible.
\end{lemma}

\begin{proof}
We prove this again by induction on $N$. Assume the set $ \{S_{2^m}^{\sigma}(0), S_{2^m}^{\sigma}(1), \ldots, S_{2^m}^{\sigma}(N-1) \} $ is admissible. We need to show that $S_{2^m}^{\sigma}(N)$ is among the minima suggested by the algorithm.
This can be shown along the same lines as Theorem \ref{thm:vdc}.
\end{proof}
Together with Corollary \ref{cor:family} this completes the proof of Theorem \ref{thm:main}.




\section{Proof of Theorem \ref{thm:main2}}
\label{sec:main2}

Throughout this section $f$ and $X$ are as in the statement of the theorem. Since $f$ has mean value 0 we have that
$$ \int_0^1 \ \sum_{k=1}^{n-1} f(|x-x_k|) \ dx=0.  $$
We need two ingredients for the proof of Theorem \ref{thm:main2}. We start with a technical observation.

\begin{lemma} We have
$$ \sum_{k, \ell = 1}^{n}{f( | x_k - x_{\ell} |)} \leq n f(0).$$
\end{lemma}

\begin{proof} We write
\begin{align*}
\sum_{k, \ell = 1}^{n}{f( | x_k - x_{\ell} |)} = n f(0) + 2 \sum_{k, \ell = 1\atop k < \ell}^{N}{f( | x_k - x_{\ell} |)}
\end{align*}
which can be rewritten as
\begin{align*}
 \sum_{k, \ell = 1\atop k < \ell}^{n}{f( | x_k - x_{\ell} | )} = \sum_{\ell=2}^{n}{ \sum_{k=1}^{\ell-1} f(| x_k - x_{\ell} | )}.
\end{align*}
Next we use the definition of the greedy algorithm. We have that
$$  \sum_{k=1}^{\ell-1} f( | x_k - x_{\ell} | ) = \min_x \sum_{k=1}^{\ell-1} f( | x-x_k |)  \leq \int_0^1 \sum_{k=1}^{\ell -1} f(|x-x_k |) \ dx= 0.$$
\end{proof}

The second ingredient is LeVeque's inequality; see e.g. \cite[Chapter 2, Theorem 2.4]{KN}.
\begin{lemma} The discrepancy $D_N$ of the finite set of points $\{ x_1, \ldots, x_N \} \subset [0,1)$ satisfies
 $$ D_N \leq \left(\frac{6}{\pi^2} \sum_{k=1}^{\infty} \frac{1}{k^2} \left|\frac{1}{N}\sum_{m = 1}^{N}e^{2\pi i k(x_m)}\right|^2 \right)^{1/3}.$$
\end{lemma}

\begin{proof}[Proof of Theorem \ref{thm:main2}.] We can rewrite the double sum using Fourier analysis as
\begin{align*}
 \sum_{m, n = 1}^{N}{f(|x_m - x_n|)} &=  \sum_{k \in \mathbb{Z} \atop k \neq 0}{\widehat{f}(k) \sum_{m, n = 1}^{N}e^{2\pi i k(x_m - x_n)}} \\
 &= \sum_{k \in \mathbb{Z} \atop k \neq 0}{\widehat{f}(k) \left(\sum_{m = 1}^{N}e^{2\pi i k(x_m)}\right) \left(\sum_{m = 1}^{N}e^{2\pi i k(-x_m)}\right)  }\\
 &=  \sum_{k \in \mathbb{Z} \atop k \neq 0}{\widehat{f}(k) \left(\sum_{m = 1}^{N}e^{2\pi i k(x_m)}\right) \overline{\left(\sum_{m = 1}^{N}e^{2\pi i k(x_m)}\right) } }\\
 &=  \sum_{k \in \mathbb{Z} \atop k \neq 0}{\widehat{f}(k) \left|\sum_{m = 1}^{N}e^{2\pi i k(x_m)}\right|^2  }\\
 &\geq N^2   \sum_{k \in \mathbb{Z} \atop k \neq 0}{\frac{c}{k^2} \left|\frac{1}{N}\sum_{m = 1}^{N}e^{2\pi i k(x_m)}\right|^2  }.
 \end{align*}
 Coupled with the Lemma above, we obtain
 $$ \sum_{k \in \mathbb{Z} \atop k \neq 0}{\frac{c}{k^2} \left|\frac{1}{N}\sum_{m = 1}^{N}e^{2\pi i k(x_m)}\right|^2  } \leq \frac{f(0)}{N}.$$
LeVeque's inequality now implies the desired result.
\end{proof}

\section*{Acknowledgements}
I would like to thank Stefan Steinerberger for inspiration and help and for sharing his vast knowledge of the related literature. He inspired my work on the main problem and suggested to add Theorem \ref{thm:main2} to an earlier version of the manuscript. Furthermore, I would like to thank Joaquim Ortega-Cerd\`{a} to bring Leja sequences to my attention.



\begin{thebibliography}{20}

\bibitem{leja1} L. Bialas-Ciez and J.-P. Calvi, Pseudo Leja sequences. Annali di Matematica {\bf 191} (2012), 55--75.

\bibitem{bl}
X.~Blanc and M.~Lewin.
\newblock{The crystallization conjecture: a review}.
\newblock{EMS Surv.~Math.~Sci.}, 2(2):225--306, (2015).

\bibitem{saffneu3}  S. Borodachov, D. Hardin and E. Saff, Asymptotics for discrete weighted minimal Riesz energy problems on rectifiable sets. Trans. Amer. Math. Soc. 360 (2008), no. 3, 1559--1580.

\bibitem{brauchart} J. Brauchart, Optimal discrete Riesz energy and Discrepancy, Uniform Distribution Theory 6, 207--220 (2011).


\bibitem{brauch} J. Brauchart and P. Grabner, 
Distributing many points on spheres: minimal energy and designs. 
J. Complexity 31 (2015), no. 3, 293--326. 

\bibitem{brauchart3} J. Brauchart, D. Hardin and E. Saff, Discrete energy asymptotics on a Riemannian circle. Unif. Distrib. Theory {\bf 7} (2012), no. 1, 77-108.


\bibitem{cf93} H. Chaix and H. Faure: \emph{Discr\'{e}pance et diaphonie en dimension un}, Acta Arith. {\bf 63} (1993), 103--141.

\bibitem{many} M. Calef, W. Griffiths, A. Schulz, C. Fichtl and D. Hardin, 
Observed asymptotic differences in energies of stable and minimal point configurations on $\mathbb{S}^2$ and the role of defects.
J. Math. Phys. 54 (2013), no. 10, 101901, 20 pp. 

\bibitem{cohn}
H.~Cohn and A.~Kumar.
\newblock{Universally Optimal Distribution of Points on Spheres}.
\newblock{\em Journal of the American Mathematical Society}, 20(1):99--148, (2007).

\bibitem{DP} J. Dick and F. Pillichshammer: Digital Nets and Sequences. Discrepancy Theory and Quasi-Monte Carlo Integration, Cambridge University Press, Cambridge, 2010.

\bibitem{faure1} H. Faure: \emph{Discr\'{e}pance de suites associ\'{e}es \`{a} un syst\`{e}me de num\'{e}ration (en dimension un)}, Bull. Soc. Math. France {\bf 109} (1981), 143--182.

\bibitem{survey1} H. Faure, P. Kritzer and F. Pillichshammer: \emph{From van der Corput to modern constructions of sequences for quasi-Monte Carlo rules.} Indag. Math., New Ser. {\bf 26} (2015), 760--822.

\bibitem{saff2} A. B. J. Kuijlaars and E. B. Saff, Distributing many points on a sphere, Mathematical Intelligencer 19, 5--11, (1997).

\bibitem{saff} A. B. J. Kuijlaars and E. B. Saff, Asymptotics for minimal discrete energy on the
sphere, Trans. Amer. Math. Soc. 350, 523--538, (1998).

\bibitem{KN} L. Kuipers and H. Niederreiter: Uniform Distribution of Sequences,
Wiley, New York, New York, 1974.

\bibitem{leja2} F. Leja, Sur certaines suites li\'ees aux ensembles plans et leur application a la repr\'esentation conforme. Ann. Polon. Math. {\bf 4} (1957), 8--13.

\bibitem{patop18} F. Pausinger and A. Topuzo\u{g}lu: \emph{On the discrepancy of two families of permuted van der Corput sequences.} Unif. Distrib. Theory {\bf 13} (2018), no. 1, 47--64.

\bibitem{pausinger1} F. Pausinger: \emph{On the intriguing search for good permutations}. Unif. Distrib. Theory {\bf 14} (2019), no. 1, 53-86.

\bibitem{radin}
C.~Radin.
\newblock{The ground state for soft disks}.
\newblock{\em Journal of Statistical Physics}, 26(2):365--373, (1981).

\bibitem{sch72} W.M. Schmidt: \emph{Irregularities of distribution VII}. Acta Arith. {\bf 21} (1972), 45--50.

\bibitem{schwartz} R. Schwartz, The 5-Electron Case of Thomson's Problem, Experimental Math 22, p. 157--186, (2013).

\bibitem{smale} S. Smale, Mathematical problems for the next century. Math. Intelligencer 20 (1998), no. 2, 7--15. 

\bibitem{stefan1} S. Steinerberger, \emph{Dynamically defined sequences with small discrepancy.} Preprint, arXiv:1902.03269.

\bibitem{stefan2} S. Steinerberger, \emph{A nonlocal functional promoting low-discrepancy point sets.} Preprint, arXiv:1902.00441.


\bibitem{theil}
F.~Theil.
\newblock{A proof of crystallization in two dimensions}.
\newblock{\em Communications in Mathematical Physics}, 262(1):209--236, (2006). 


\bibitem{thomson} J. J. Thomson, On the Structure of the Atom: an Investigation of the Stability and Periods of Oscillation of
a number of Corpuscles arranged at equal intervals around the Circumference of a Circle; with Application of
the Results to the Theory of Atomic Structure, Philosophical Magazine Series 6, Volume 7, Number 39, pp.
237--265, March 1904.

\bibitem{vdc35a} {\sc J.G. Van der Corput}: \emph{Verteilungsfunktionen. I,} Proc. Kon. Ned. Akad. v. Wetensch. {\bf 38} (1935), 813--821.

\bibitem{vdc35b} {\sc J.G. Van der Corput}: \emph{Verteilungsfunktionen. II,} Proc. Kon. Ned. Akad. v. Wetensch. {\bf 38} (1935), 1058--1066.


\end{thebibliography}
\end{document}